\documentclass[12pt]{article}
\usepackage{natbib}
\usepackage{yaonan}

\newcommand{\ExAnteRelaxation}{\textsf{Ex Ante Relaxation}}

\newcommand{\EAR}{{\sf EAR}}

\providecommand{\keywords}[1]
{
  \small	
  \textbf{Keywords---} #1
}

\providecommand{\JEL}[1]
{
  \small	
  \textbf{JEL---} #1
}

\title{Anonymous Pricing in Large Markets\thanks{This work has evolved over a couple of
years and has benefited from comments from many colleagues.
We are especially grateful to Yiding Feng, Shunhua Jiang, Pinyan Lu, Tim Roughgarden, and Hengjie Zhang for discussions that inspired this work.
Yaonan Jin acknowledges financial support from the HKUST Start-up Grant, and Yingkai Li acknowledges financial support from the NUS Start-up Grant.}}
\author{
Yaonan Jin\thanks{Hong Kong University of Science and Technology. Email: {\tt jinyaonan1996@gmail.com}}
\and
Yingkai Li\thanks{National University of Singapore. Email: {\tt yk.li@nus.edu.sg}}
}
\date{}

\begin{document}
\newcommand{\setsize}[1]{{\left|#1\right|}}

\newcommand{\floor}[1]{
{\lfloor {#1} \rfloor}
}
\newcommand{\bigfloor}[1]{
{\left\lfloor {#1} \right\rfloor}
}

%
% probability stuff.
%
\newcommand{\given}{\,|\,}

% resizing brackets 
\newcommand{\prob}[2][]{\text{\bf Pr}\ifthenelse{\not\equal{}{#1}}{_{#1}}{}\!\left[{\def\givenn{\middle|}#2}\right]}
\newcommand{\expect}[2][]{\text{\bf E}\ifthenelse{\not\equal{}{#1}}{_{#1}}{}\!\left[{\def\givenn{\middle|}#2}\right]}

% brackets configured with \tparen
\newcommand{\tparen}{\big}
\newcommand{\tprob}[2][]{\text{\bf Pr}\ifthenelse{\not\equal{}{#1}}{_{#1}}{}\tparen[{\def\given{\tparen|}#2}\tparen]}
\newcommand{\texpect}[2][]{\text{\bf E}\ifthenelse{\not\equal{}{#1}}{_{#1}}{}\tparen[{\def\given{\tparen|}#2}\tparen]}

% brackets do not resize.
\newcommand{\sprob}[2][]{\text{\bf Pr}\ifthenelse{\not\equal{}{#1}}{_{#1}}{}[#2]}
\newcommand{\sexpect}[2][]{\text{\bf E}\ifthenelse{\not\equal{}{#1}}{_{#1}}{}[#2]}

% brackets
\newcommand{\rbr}[1]{\left(\,#1\,\right)}
\newcommand{\sbr}[1]{\left[\,#1\,\right]}
\newcommand{\cbr}[1]{\left\{\,#1\,\right\}}

\newcommand{\suchthat}{\,:\,}

\newcommand{\partialx}[2][]{{\tfrac{\partial #1}{\partial #2}}}
\newcommand{\nicepartialx}[2][]{{\nicefrac{\partial #1}{\partial #2}}}
\newcommand{\dd}{\,{\mathrm d}}
\newcommand{\ddx}[2][]{{\tfrac{\dd #1}{\dd #2}}}
\newcommand{\niceddx}[2][]{{\nicefrac{\dd #1}{\dd #2}}}
\newcommand{\grad}{\nabla}

\newcommand{\symdiff}{\triangle}

\newcommand{\abs}[1]{\left| #1 \right |}

\newcommand{\reals}{{\mathbb R}}
\newcommand{\posreals}{\reals_+}
\newcommand{\indicate}[1]{\mathbf{1}\left[\,#1\,\right]}

\newcommand{\primed}{^{\dagger}}

\newcommand{\val}{v}
\newcommand{\alloc}{x}
\newcommand{\pay}{p}
\newcommand{\util}{u}
\newcommand{\dist}{F}
\newcommand{\density}{f}
\newcommand{\quant}{q}
\newcommand{\rev}{R}
\newcommand{\iron}{\bar{\rev}}
\newcommand{\mquant}{\quant^*}
\newcommand{\mval}{\val^*}
\newcommand{\mrev}{r^*}
\newcommand{\opt}{{\rm OPT}}
\newcommand{\tri}{{\rm Tri}}
\newcommand{\up}{{\rm AP}}
\newcommand{\spp}{{\rm SPP}}
\newcommand{\os}{D}
\newcommand{\virtual}{\phi}
\newcommand{\virtualce}{\virtual_{\rm CE}}
\newcommand{\mech}{M}
\newcommand{\apname}{anonymous}
\newcommand{\aapname}{an anonymous}
\newcommand{\Apname}{Anonymous}
\newcommand{\sppname}{sequential posted}
\newcommand{\Sppname}{Sequential Posted}

\maketitle
\begin{abstract}
We study revenue maximization when a seller offers $k$ identical units to ex ante heterogeneous, unit-demand buyers. While \apname~pricing can be $\Theta(\log k)$ worse than optimal in general multi-unit environments, we show that this pessimism disappears in large markets, where no single buyer accounts for a non-negligible share of optimal revenue. Under (quasi-)regularity, \apname~pricing achieves a $2+O(1/\sqrt{k})$ approximation to the optimal mechanism; the worst-case ratio is maximized at about $2.47$ when $k=1$ and converges to $2$ as $k$ grows. This indicates that the gains from third-degree price discrimination are mild in large markets. 
\end{abstract}

\noindent\keywords{\apname~pricing, third-degree price discrimination, large markets, approximation}

\noindent\JEL{D44,D47,D83}
\thispagestyle{empty}

\newpage
\setcounter{page}{1}

\section{Introduction}
\label{sec:intro}

In single-item auctions, a remarkably robust insight has emerged from recent theory: simple \apname~pricing mechanisms can perform nearly as well as the fully optimal mechanism.
In particular, \citet{AHNPY19} show that under standard regularity conditions, \apname~pricing achieves a constant-factor approximation to the optimal revenue despite the heterogeneity of their preferences.\footnote{The tight approximation ratio is derived in \citet{JLQTX19,JLTX20} and is approximately 2.62.}
This result provides a compelling theoretical justification for the widespread use of \apname~pricing in practice, where simplicity, transparency, and ease of implementation are often paramount.

However, this positive conclusion hinges critically on the single-item environment.
When the seller supplies multiple identical units to agents with unit demands, the performance of \apname~pricing deteriorates sharply.
As shown by \citet{HR09,JJLZ22}, in $k$-unit auctions, \apname~pricing can incur a worst-case multiplicative revenue loss on the order of $\Theta(\log k)$ relative to the optimal mechanism.
This gap grows unbounded as the supply expands, suggesting that the robustness of \apname~pricing is fundamentally limited in multi-unit environments.

This limitation is not merely theoretical.
In many economically important markets, sellers routinely face large and variable supplies rather than a single indivisible object.
Retailers sell inventories consisting of millions of identical units; digital marketplaces distribute standardized goods to large populations of buyers; platforms allocate advertising impressions or data access rights at scale; and technology firms license software to consumers in large volumes.
In such settings, the supply parameter $k$ is naturally large, precisely where the worst-case inefficiency of \apname~pricing becomes most pronounced.

In this paper, we recover the approximate optimality of \apname~pricing in multi-unit auctions with a large supply of $k$, under the large market assumption. 
To provide intuition about the assumption and our main approximation result, we would like to provide a simple illustration of the loss of \apname~pricing in multi-unit auctions.
\begin{example}
Consider selling $k$ units to $k$ agents. Each agent $i$ has a value that is deterministic at $\frac{1}{i}$. 
The optimal mechanism sells a unit to agent $i$ at a price of $\frac{1}{i}$, with a revenue of $\sum_{i\in[k]} \frac{1}{i} = \Theta(\log k)$. 
However, for \apname~pricing mechanisms, the number of sales at any price $\frac{1}{i}$ is $i$, with a revenue of $1$. 
\end{example}

In this example, the large gap arises because revenue is extremely concentrated among a vanishingly small set of agents.
A single agent contributes revenue of order one, while the average revenue per unit is only $O(\log k / k)$.
As a result, the optimal mechanism relies on finely tailored prices that extract surplus from a few high-value agents, whereas \apname~pricing is forced to trade off revenue from these agents against sales to a large mass of low-value buyers.

Such extreme revenue concentration is precisely what our large-market assumption rules out.
In competitive markets with many potential buyers, no single agent—or small subset of agents—accounts for a non-negligible share of total revenue.
Instead, revenue is dispersed across many agents, and marginal contributions are small relative to aggregate market size.
Under this natural assumption, the inefficiency illustrated by the example is mitigated, and \apname~pricing mechanisms recover constant-factor approximations even when supply is large. 
Moreover, our results extend to environments with quasi-regular distributions \citep{feng2025beyond}, which is a generalization of the standard regularity assumption \citep{M81,BR89}.
This provides a theoretical explanation for why simple \apname~pricing mechanisms continue to perform well in large markets despite pessimistic worst-case bounds.

\subsection{Results and Techniques}
Our paper considers the problem of selling $k$ units of identical items to $n$ ex ante heterogeneous agents with unit demands for the items. 
The goal of the seller is to maximize the expected revenue. 
We assume that the value distributions of the agents are (quasi-)regular, as in \citet{M81,AHNPY19,feng2025beyond}. 

A crucial assumption in our model is the large market assumption. 
Formally, a market is $\epsilon$-large for $\epsilon>0$ if, for any agent $i\in[n]$, the optimal revenue from a single agent $i$ cannot exceed $\epsilon$ fraction of the optimal revenue from all $n$ agents. 
This assumption captures the competitiveness in the market, and in our paper, we will mainly focus on the large market limit where $\epsilon\to 0$. 

\medskip

\noindent\textbf{Main Theorem.} \emph{Under large market and (quasi-)regularity assumptions, for $k$-unit auctions, \apname~pricing achieves a $(2+\frac{O(1)}{\sqrt{k}})$-approximation to the optimal revenue. 
Moreover, the worst case approximation ratio is maximized at $2.47$ when $k=1$, and is at least $2$ for any $k\geq 2$.}
\medskip

The proof relies on a reduction to the truncated Pareto distribution as a worst case scenario and a novel methodology for quantifying the order statistics which is used for bounding the worst case gap. 

\medskip\noindent\textbf{Reduction to worst case instances.} We show that in $k$-unit auctions, under large market and regularity assumptions, we can decompose the value distribution of any single agent into a group of smaller agents, each endowed with a truncated Pareto distribution. 
Note that truncated Pareto distributions are often referred to as triangular distributions in the literature \citep[e.g.,][]{AHNPY19}, where the revenue curves of those distributions take a triangular shape. 

In our construction, we create truncated Pareto distributions such that the revenue contribution of these agents remains almost the same under both optimal revenue mechanisms and \apname~pricing mechanisms. 
Moreover, by splitting the agents into smaller ones, the large market assumption remains valid as the revenue contribution from each individual agent can only decrease after the split. 
Therefore, we show that the worst case scenario takes the form of truncated Pareto distributions (\cref{prop:triangle_is_worst_case}). 

One benefit of the worst case reduction is that each truncated Pareto distribution is determined only by two parameters, and the revenue optimal mechanism for those distributions is a sequential posted pricing mechanism. 
This allows us to represent both the optimal revenue and the \apname~pricing mechanism as a function of the order statistics, similar to \citet{JLTX20,JJLZ22}. 

\medskip\noindent\textbf{Order statistics in large markets.} 
In our paper, we show that under the large market assumption, all higher order statistics can essentially be pinned down by the first order statistics (\cref{prop:order_statistics}). 
The key observation is that, at any fixed price, the event that a given buyer is willing to pay that price is typically rare in a large market. The number of buyers who exceed the price is therefore the sum of many ``small'' independent chances. 
In such settings, a standard large-market heuristic is that this count behaves almost like a simple random counting process determined by a single aggregate parameter—roughly, the overall likelihood that nobody exceeds the price. Once that aggregate likelihood is known, the probability of seeing at most one exceedance, at most two exceedances, and so on (which are exactly the higher-order statistics) are essentially pinned down. 
The only reason this heuristic is not exact is that there can be minor dependencies created by ``double-counting'' the same buyer when one approximates distinct exceedances by repeated sampling, but in large markets, each buyer's chance of exceeding the price is so small that such collisions are negligible. This is why, under our large-market assumption, the entire collection of higher-order statistics is well-approximated by a function of the first-order statistic alone.

\medskip

With the characterization of the order statistics in large markets, the revenue from \apname~pricing, the optimal revenue, and the worst case gap can be uniquely determined by the distribution over first order statistics. 
This allows us to derive the worst case approximation in closed form. 
Our main result then follows from bounding the asymptotic order of the worst case approximation. 

\subsection{Related Work}
\label{subsec:literature}
The investigation of the approximation guarantees of simple mechanisms has a long history in the algorithmic game theory community, dating back to \citet{HR09}. The literature has focused on various simple mechanisms such as second price auctions \citep[][]{BK96}, sequential posted pricing mechanisms \citep[e.g.,][]{Y11}, \apname~pricing \citep[e.g.,][]{AHNPY19}, and so on. 
Follow up work has extended the results to multi-unit environments \citep[e.g.,][]{DFK16,JJLZ22}, multi item auctions \citep[e.g.,][]{CHMS10,LY13,HN17,CDW21,CZ17,JL24}, non-linear utilities such as budget constraints \citep[e.g.,][]{cheng2018simple,feng2019optimal,feng2023simple},
or even beyond auction environments \citep[e.g.,][]{dutting2019simple}. 
Beyond the theoretical perspectives, these simple mechanisms often show good performance based on real-world data, as suggested by \citet{coey2021scalable,dellavigna2019uniform}.

Our paper focuses on \apname~pricing mechanisms. In single-item environments, \citet{AHNPY19} show that \apname~pricing guarantees at least a $e \approx 2.72$-approximation to the optimal. This approximation ratio is improved to $2.62$ and is shown to be tight in \citet{JLQTX19,JLTX20}.
Relatedly, \citet{bergemann2022third} show that \apname~pricing achieves a 2-approximation to the optimal if all agents share a common support for their value distributions. 
\citet{HR09,JJLZ22} show that in $k$-unit auctions, the tight approximation of \apname~pricing is of the order of $\Theta(\log k)$. 
The main contribution of our paper, compared to the literature, is to focus on the large market assumption and show that \apname~pricing achieves a constant approximation even in multi-unit auctions. 

The large market assumption is widely adopted in various settings, such as budget feasible procurement auctions \citep[e.g.,][]{AGN18,JT21}, quality of equilibria \citep[e.g.,][]{S01,CT16,FILRS16}, and (generalized) online matching \citep[e.g.,][]{MSVV07}.
% or prophet inequalities \citep[e.g.,][]{AEEHK17, LPSS21, BCDEFV24}.
Among those papers, the one that is most relevant for us is \citet{S01}, which focuses on the welfare objective and shows that simple auctions achieve a near optimal welfare guarantee under the large market assumption.

\section{Model}
\label{sec:prelim}

A principal sells $k\geq 1$ units of identical items to $n$ different agents to maximize the expected revenue. 
All agents have unit-demand, and for each agent $i$ with value $\val_i\in\reals_+$, his utility for getting an item with probability $\alloc_i\in[0,1]$ and paying a price $\pay_i\in\reals$ is 
\begin{align*}
\util_i(\alloc_i,\pay_i;\val_i) = \val_i\cdot\alloc_i - \pay_i.
\end{align*}
The private value $\val_i$ for each agent $i$ is drawn independently from the distribution $\dist_i$. 
The prior distribution $\dist=\times_{i\in[n]} \dist_i$ is publicly known. 
We also use $\dist$ to denote the auction design instance. 
The objective of the principal is to maximize the total revenue $\sum_i \pay_i$. 

\paragraph{Mechanisms}
By the revelation principle \citep{M81}, it is without loss of generality to focus on revelation mechanisms. 
A revelation mechanism $\mech=(\alloc,\pay)$ is determined by the allocation rule $\alloc: \reals^n \to [0,1]^n$ 
such that $\sum_i \alloc_i(\val) \leq k$,
and the payment rule $\pay: \reals^n \to \reals^n$.
A revelation mechanism $\mech$ must satisfy the incentive compatibility (IC) and individual rationality (IR) constraints: 
for any agent $i$ and any value $\val_i,\val'_i$,
\begin{align}
\expect[\val_{-i}\sim F_{-i}]{\val_i\cdot\alloc_i(\val_i,\val_{-i}) - \pay_i(\val_i,\val_{-i})} 
&\geq \expect[\val_{-i}\sim F_{-i}]{\val_i\cdot\alloc_i(\val'_i,\val_{-i}) - \pay_i(\val_i',\val_{-i})}, \tag{IC} \\
\expect[\val_{-i}\sim F_{-i}]{\val_i\cdot\alloc_i(\val_i,\val_{-i}) - \pay_i(\val_i,\val_{-i})} 
&\geq 0. \tag{IR}
\end{align}
Given any incentive compatible and individually rational mechanism $\mech$, we use $\mech(\dist)$ to denote the expected revenue of $\mech$ given distribution $\dist$.
We denote the optimal revenue as 
\begin{align*}
\opt(\dist) = \max_{\mech \text{ is IC, IR}} \mech(\dist).
\end{align*}

In our paper, there are two special families of mechanisms that are of particular interest, besides the revenue optimal mechanism. 

\begin{itemize}
\item \textbf{\Sppname~Pricing.}
In a \sppname~pricing mechanism, the principal sets a personalized price $\pay_i$ for each agent $i$. 
The principal offers $\pay_i$ to agent $i$ as a take-or-leave-it offer in decreasing order of $p_i$ until the supply runs out. 
Let $\spp(\pay,\dist)$ be the expected revenue from \sppname pricing given the price vector $\pay$.

\item \textbf{\Apname~Pricing.}
In \aapname~pricing mechanism, the principal sets \aapname~price $\pay$ for all agents. 
The principal offers $\pay$ to each agent $i$ as a take-or-leave-it offer in an arbitrary order until the supply runs out. 
Let $\up(\pay,\dist)$ be the expected revenue from \apname~pricing given the \apname~price $\pay$.
\end{itemize}

We use $\spp(\dist) = \max_{\pay} \spp(\pay,\dist)$ and $\up(\dist) = \max_{\pay} \up(\pay,\dist)$ to denote the optimal revenue that can be attained from the family of \sppname~pricing mechanisms and \apname~pricing mechanisms respectively. 

\paragraph{Revenue Curves} A useful technique in our paper is the analysis through revenue curves in quantile space. 
Specifically, for each agent $i$ with distribution $\dist_i$, 
the quantile for value $\val_i$ is defined as $\quant_i(\val_i) = \Pr_{z\sim\dist_i}[z\geq \val_i]$. 
Correspondingly, we can define the value corresponding to quantile $\quant_i$ as $\val_i(\quant_i) = \inf \{\val: \dist_i(\val) \geq 1-\quant_i\}$.
Based on the definition of quantiles, we can then define the revenue curve in quantile space as 
\begin{align*}
\rev_i(\quant_i) = \quant_i \cdot \val_i(\quant_i). 
\end{align*}

Given the definition of revenue curves, let $\quant^*_i$ be the monopoly quantile that maximizes the revenue, i.e., 
$\quant^*_i = \argmax_{\quant} \rev_i(\quant)$. 
When there are multiple optimizers, we break ties by choosing the smallest $\quant^*_i$.

\begin{definition}[Regularity]
\label{asp:regular}
The distributions $\dist_i$ are regular for all agents $i$. 
That is, the revenue curve $\rev_i(\quant_i)$ is concave in $\quant_i$ for all agents $i$. 
\end{definition}
The regularity assumption here is identical to the regularity assumption in \citet{M81} where $\phi_i(\val_i) = \val_i - \frac{1-\dist_i(\val_i)}{\density_i(\val_i)}$ is non-decreasing in $\val_i$. 
The representation in \cref{asp:regular} is introduced in \citet{BR89}, which can be interpreted as the marginal revenue for serving agents being weakly decreasing in the quantiles. 

\begin{definition}[Quasi-regular Distributions]
A distribution $\dist$ is quasi-regular if the conditional expected virtual value function $\virtualce(\val) = \expect{\virtual(\hat{\val})\given \hat{\val}\leq \val}$ is weakly increasing in $\val$. 
\end{definition}

\paragraph{Large Markets}
We focus on large market instances. That is, none of the individual agents can contribute to a significant fraction of the optimal revenue. 

\begin{definition}[$\epsilon$-large Markets]
For any $\epsilon\in (0,1]$, the instance $\dist$ has a $\epsilon$-large market if for any agent $i$, $\rev_i(\mquant_i) \leq \epsilon\cdot \opt(\dist)$.
\end{definition}
In this paper, we focus on the \emph{large market assumption}, where we consider an $\epsilon$-large market in the limit of $\epsilon\to 0$.

\section{Approximation of \Apname~Pricing}
\label{sec:uniform}

This section establishes our central result on the performance of \apname~pricing mechanisms relative to the revenue-optimal auction in multi-unit settings with large market assumptions. While previous work has demonstrated that \apname~pricing can suffer revenue losses on the order of $\Theta(\log k)$ in the worst case for $k$-unit auctions \citep{HR09,JJLZ22}, we show that this pessimistic bound does not hold under the large market assumption. Our analysis reveals that as markets grow large, the theoretical justification for \apname~pricing mechanisms becomes substantially stronger even when multiple units are available for sale.

We begin by presenting our main theoretical contribution, which characterizes the worst-case approximation ratio between \apname~pricing and the optimal mechanism. Throughout this section, we maintain the quasi-regularity assumption on value distributions together with the $\epsilon$-large market condition introduced in \cref{sec:prelim}, considering the limiting case as $\epsilon \rightarrow 0$.

\begin{theorem}
\label{thm:tight_approx}
For quasi-regular distributions that satisfy the large market assumption, the worst-case approximation ratio of \apname~pricing to the optimal revenue is 
\begin{align*}
\max_{\dist} \frac{\opt(\dist)}{\up(\dist)} = 2+\frac{O(1)}{\sqrt{k}}.
\end{align*}
Moreover, the worst case gap is maximized at $2.47$ when $k=1$. 
\end{theorem}

This theorem establishes that \apname~pricing mechanisms achieve a constant-factor approximation to the optimal revenue in large markets, with an approximation ratio of $2 + O(1/\sqrt{k})$. Most notably, this bound approaches 2 asymptotically as the number of units grows large, while being maximized at approximately 2.47 when selling a single unit.

This result represents a fundamental departure from worst-case analyses without the large market assumption. In standard multi-unit auction settings, the revenue gap between \apname~pricing and the optimal mechanism grows as $\Theta(\log k)$ with the number of units, suggesting increasingly poor performance of simple pricing mechanisms as markets scale. To formalize this contrast, we recall the following result from the literature:

\begin{proposition}[\citealp{JJLZ22}]
\label{prop:loose_gap}
For any regular distribution over values, the worst case gap between the optimal revenue and uniform pricing is $\Theta(\log k)$. 
\end{proposition}

We extend this result to the broader class of quasi-regular distributions:

\begin{corollary}
\label{cor:loose_gap}
For any quasi-regular distribution over values, the worst case gap between the optimal revenue and uniform pricing is $\Theta(\log k)$. 
\end{corollary}

Without the large market assumption, this logarithmic gap grows unboundedly with supply, suggesting that \apname~pricing performs poorly in large-scale markets. It is precisely this pessimistic conclusion that our \cref{thm:tight_approx} overturns when we restrict attention to large markets.

For our analysis, we adopt the normalization that the optimal revenue from \apname~pricing equals one. Under this normalization and given \cref{cor:loose_gap}, the optimal revenue cannot exceed $\log k$. This normalization facilitates a key property of agent distributions in large markets that will prove instrumental in our analysis. 
Let $\gamma$ be the constant in \cref{cor:loose_gap} such that the worst case gap between the optimal revenue and uniform pricing is at most $\gamma \cdot\log k$. 

\begin{corollary}\label{cor:large_cdf}
For any $\pay\geq \frac{1}{k}$, in $\epsilon$-large markets with quasi-regular distributions, we have $\dist_i(\pay) \geq 1 - \epsilon \cdot\gamma k\log k$.
\end{corollary}

This corollary confirms that in large markets, agents' values are sufficiently dispersed such that the probability of any particular agent's value exceeding a reasonable price threshold is uniformly low. This property reflects the fundamental economic feature of large markets that no individual agent's preferences dominate market outcomes.

We now proceed to develop the technical machinery needed to establish our main result. The proof of \cref{thm:tight_approx} follows from a synthesis of two key components that we develop in subsequent subsections. First, in \cref{subsec:reduction}, we show that worst-case instances can be reduced to environments where agents' value distributions take a specific triangular form (\cref{prop:triangle_is_worst_case}). This reduction preserves the essential features of the revenue optimization problem while significantly simplifying our analysis. Second, in \cref{subsec:order_statistics}, we develop novel techniques for characterizing order statistics in large markets, demonstrating that higher-order statistics can be precisely approximated by functions of the first-order statistic alone (\cref{prop:order_statistics}). Together, in \cref{subsec:bound_calculation,subsec:universal_bound}, these technical foundations enable us to derive tight bounds on the revenue gap between \apname~pricing and the optimal mechanism, ultimately establishing the approximation guarantee stated in \cref{thm:tight_approx}.

\subsection{Reduction to Triangular Instances}
In this section, we provide the formal reduction, showing that in large markets, the worst case gap occurs at triangular distributions. 
\label{subsec:reduction}
\begin{definition}[Triangular Distributions]
For any $\mquant\in[0,1]$ and any $\mrev > 0$,
a triangular distribution $\tri_{\mrev,\mquant}$ has a cumulative distribution function\footnote{In the case with $\mquant=0$, by slightly abusing the notations, we view $\frac{\mrev}{\mquant}$ as $+\infty$ and $\val\geq \frac{\mrev}{\mquant}$ never holds.} 
\begin{align*}
\tri_{\mrev,\mquant}(\val) = \begin{cases}
1 & \val\geq \frac{\mrev}{\mquant};\\
1-\frac{1}{1+\frac{\val}{\mrev}(1-\mquant)} & \val \in [0,\frac{\mrev}{\mquant}).
\end{cases}
\end{align*}
\end{definition}

Triangular distributions are special cases of regular distributions, where there is a point mass at the highest value, and the virtual value for all lower values is a negative constant.
\citet{M81} shows that the revenue optimal mechanism only sells the items to agents with the highest non-negative virtual values. 
This corresponds to posting a price at the highest value for each agent. 
\begin{lemma}\label{lem:spp_for_triangle}
For triangular distributions, the revenue-optimal mechanism is \sppname\ pricing. 
\end{lemma}

In the next proposition, we show that given any quasi-regular distributions that satisfy the large market assumption, 
we can construct triangular distributions that also satisfy the large market assumption while almost preserving the approximation guarantee. 

\begin{proposition}\label{prop:triangle_is_worst_case}
For any $k \geq 1, \epsilon \in (0, \frac{1}{\gamma k\log k})$ and any instance $\dist$ with quasi-regular distributions that satisfies the $\epsilon$-large market assumption, there exists another instance $\hat{\dist}$ with triangular distributions that also satisfies the $\epsilon$-large market assumption, and the approximation ratio of \apname~pricing decreases by at most 
$\frac{1+\epsilon\cdot \gamma k\log k}{(1-\epsilon\cdot \gamma k\log k)^2}$. 
\end{proposition}

In the limiting case where $\epsilon\to 0$, triangular distributions are the exact worst case distributions that maximize the approximation ratio of \apname~pricing.

To prove \cref{prop:triangle_is_worst_case}, we first reduce the worst case distribution to those where the negative virtual value is a constant. 
Specifically, we say that a distribution $\dist$ is a constant negative virtual value distribution if its virtual value function satisfies the following condition: there exists a constant $z < 0$ such that whenever $\virtual(\val) \leq 0$, it holds that $\virtual(\val) = z$.

\begin{lemma}\label{lem:const_negative_virtual}
For any instance $\dist$ with quasi-regular distributions that satisfies the large market assumption, there exists another instance $\hat{\dist}$ with constant negative virtual value distributions that also satisfies the large market assumption and a weakly larger approximation ratio. 
\end{lemma}
\begin{proof}
For any agent $i$, and any distribution $\dist_i$, 
let $\mquant$ be the monopoly quantile, i.e., the smallest quantile with non-positive virtual value.
Consider a distribution $\hat{\dist}_i$ such that $\hat{\rev}_i(\quant_i)=\rev_i(\quant_i)$ for any $\quant < \mquant$, 
and $\hat{\rev}_i(\quant_i) = \frac{1-\quant_i}{1-\mquant}$ for any $\quant\in[\mquant,1]$. 
The quasi-regularity of distribution $\dist_i$ implies that $\dist_i$
first-order stochastically dominates $\hat{\dist}_i$. 
By repeating this construction for all $i$, the expected revenue from uniform pricing weakly decreases. 

Moreover, note that in the optimal mechanism, the item is not allocated to agents with negative virtual values. 
In the above construction, the distribution over positive virtual values are the same for both $\dist$ and $\hat{\dist}$. 
Therefore, the optimal revenue remains the same. 
This implies that the large market assumption still holds, 
and the approximation ratio weakly increases. 
\end{proof}

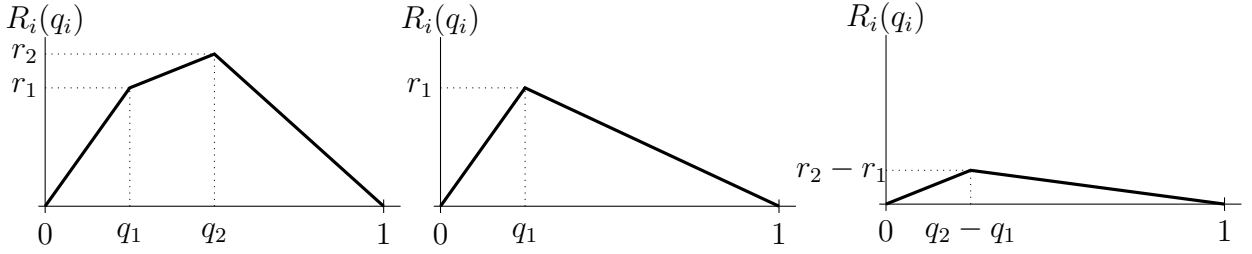
\begin{figure}[t]
\begin{flushleft}
\begin{minipage}[t]{0.31\textwidth}
\centering
\begin{tikzpicture}[scale = 0.45]

\draw (-0.2,0) -- (10.5, 0);
\draw (0, -0.2) -- (0, 5);

\draw [very thick] (0, 0) -- (2.5, 3.5) -- (5, 4.5) -- (10, 0);
\draw [dotted] (2.5,0) -- (2.5, 3.5);
\draw [dotted] (5,0) -- (5, 4.5);
\draw (2.5, -0.8) node {$\quant_1$};
\draw (5, -0.8) node {$\quant_2$};

\draw [dotted] (0,3.5) -- (2.5, 3.5);
\draw [dotted] (0,4.5) -- (5, 4.5);
\draw (-0.6,3.5) node {$r_1$};
\draw (-0.6,4.5) node {$r_2$};

\draw (0, -0.8) node {$0$};
\draw (0, 5.5) node {$\rev_i(\quant_i)$};

\draw (10, -0.2) -- (10, 0.2);
\draw (10, -0.8) node {$1$};
\end{tikzpicture}
\end{minipage}
\begin{minipage}[t]{0.31\textwidth}
\centering
\begin{tikzpicture}[scale = 0.45]

\draw (-0.2,0) -- (10.5, 0);
\draw (0, -0.2) -- (0, 5);

\draw [very thick] (0, 0) -- (2.5, 3.5) -- (10, 0);
\draw [dotted] (2.5,0) -- (2.5, 3.5);
\draw (2.5, -0.8) node {$\quant_1$};

\draw [dotted] (0,3.5) -- (2.5, 3.5);
\draw (-0.6,3.5) node {$r_1$};

\draw (0, -0.8) node {$0$};
\draw (0, 5.5) node {$\rev_i(\quant_i)$};

\draw (10, -0.2) -- (10, 0.2);
\draw (10, -0.8) node {$1$};
\end{tikzpicture}
\end{minipage}
\begin{minipage}[t]{0.36\textwidth}
\centering
\begin{tikzpicture}[scale = 0.45]

\draw (-0.2,0) -- (10.5, 0);
\draw (0, -0.2) -- (0, 5);

\draw [very thick] (0, 0) -- (2.5, 1) -- (10, 0);
\draw [dotted] (2.5,0) -- (2.5, 1);
\draw (2.5, -0.8) node {$\quant_2-\quant_1$};

\draw [dotted] (0,1) -- (2.5, 1);
\draw (-1.3,1) node {$r_2-r_1$};

\draw (0, -0.8) node {$0$};
\draw (0, 5.5) node {$\rev_i(\quant_i)$};

\draw (10, -0.2) -- (10, 0.2);
\draw (10, -0.8) node {$1$};
\end{tikzpicture}
\end{minipage}
\end{flushleft}
\caption{\label{f:decompose}
Decomposition of a revenue curve into triangular ones.}
\end{figure}

\begin{proof}[Proof of \cref{prop:triangle_is_worst_case}]
By \cref{lem:const_negative_virtual}, it is without loss of generality to focus on distributions with constant negative virtual values. 
Moreover, in the $\epsilon$-large market, it is without loss to assume that there are $N$ agents with a triangular distribution $\tri_{\frac{\epsilon}{k},\epsilon}$ in the instance, where $N\to \infty$, 
since including those agents with maximum value $\frac{1}{k}$ would not increase the revenue from \apname~pricing.
However, adding those agents ensures that in the revenue optimal mechanisms, all $k$ units can be sold at least to those agents with price $\frac{1}{k}$. 

Next, we formally introduce our reduction instance. 
Note that for quasi-regular distributions, the revenue curves are not concave in general. 
We denote the ironed revenue curve for agent $i$ as $\iron_i$, 
which is the concave hull of $\rev_i$. 
First, consider a discretization to approximate the ironed revenue curves using piecewise linear functions for all agents $i$. 
% Specifically, for any discretization grid $\delta$, 
Specifically, each discretized revenue curve can be represented by two sequences of parameters $\{r^{(0)}_i, r^{(1)}_i, \dots,r^{(Z_i)}_i\}$
and $\{\quant^{(0)}_i,\quant^{(1)}_i,\dots,\quant^{(Z_i)}_i\}$, where $\quant^{(0)}_i=0, \quant^{(Z_i)}_i = \quant^*_i$, $\quant^{(z)}_i$ is an increasing sequence and $r^{(x)}_i = \rev_i(\quant^{(z)}_i)$ for any $z\in[Z_i]$. 
% Moreover, $(1+\delta) \quant^{(z)}_i \geq \quant^{(z+1)}_i$ and $(1+\delta) \rev^{(z)}_i \geq \rev^{(z+1)}_i$ for all $z$.
We decompose the revenue curve $\rev_i$ for agent $i$ into $Z_i$ revenue curves corresponding to $Z_i$ agents with triangular distributions. 
Specifically, for each constructed agent $z\in[Z_i]$, the value distribution is $\tri_{r^{(z)}_i-r^{(z-1)}_i,\quant^{(z)}_i-\quant^{(z-1)}_i}$. 
We use $\hat{\dist}$ to denote the constructed instance. 
Note that the constructed instance $\hat{\dist}$ still satisfies the large market assumption as the monopoly revenue from each individual agent in $\hat{\dist}$ can be infinitesimally small. 

\begin{lemma}\label{lem:optimal_increase}
$\opt(\hat{\dist}) \geq (1-\epsilon\cdot \gamma k\log k)\cdot \opt(\dist)$.
\end{lemma}
\begin{proof}
We adopt the result from \citet{M81} where the expected revenue equals the expected virtual value.
To provide a lower bound on the optimal revenue for $\hat{\dist}$, we consider a restricted mechanism where at most one of the $Z_i$ agents can win an item. 
Therefore, we can view those $Z_i$ agents as a single agent whose contribution to the virtual welfare equals the maximum individual virtual welfare. 
Moreover, due to the existence of those additional $N$ agents, it is sufficient to count only the virtual welfare contribution for virtual value at least $\frac{1}{k}$. 
Letting $z_i$ be the maximum index such that the corresponding virtual value is at least $\frac{1}{k}$, i.e., 
\begin{align*}
\frac{r^{(z)}_i-r^{(z-1)}_i}{\quant^{(z)}_i-\quant^{(z-1)}_i} \geq \frac{1}{k}.
\end{align*}
Let $\alloc_i(z)$ be the interim allocation of agent $i$ when his virtual value corresponds to the interval for the decomposed agent~$z$. 
We have, 
\begin{align*}
\opt(\hat{\dist}) &\geq \sum_{i\in[n]} \sum_{z\in [z_i]} (1-\quant^{(z-1)}_i)\frac{r^{(z)}_i-r^{(z-1)}_i}{\quant^{(z)}_i-\quant^{(z-1)}_i} \cdot \alloc_i(z) \cdot (\quant^{(z)}_i-\quant^{(z-1)}_i)\\
&\geq (1-\epsilon\cdot \gamma k\log k) \cdot \sum_{i\in[n]} \sum_{z\in [z_i]} (r^{(z)}_i-r^{(z-1)}_i) \cdot \alloc_i(z)
= (1-\epsilon\cdot \gamma k\log k) \cdot \opt(\dist),
\end{align*}
where the inequality holds since in $\epsilon$-large markets, we have $\quant^{(z)}_i \leq \epsilon\cdot \gamma k\log k$. 
\end{proof}

\begin{lemma}\label{lem:uniform_pricing_decreasing}
$\up(\hat{\dist})\leq \frac{1+\epsilon\cdot \gamma k\log k}{1-\epsilon\cdot \gamma k\log k}\cdot \up(\dist)$.
\end{lemma}
\begin{proof}
For any agent $i$ and any price $p\geq \frac{1}{k}$, 
let $z_i\in[Z_i]$ be the maximum index such that 
$\frac{r^{(z_i)}_i}{\quant^{(z_i)}_i} \leq p$. 
Given price $p$, due to quasi-regularity, the probability that agent $i$ is willing to accept the price is at least 
\begin{align*}
A_i(p)\triangleq\frac{r^{(z_i)}_i}{r^{(z_i)}_i+(1-\quant^{(z_i)}_i)\cdot p} \geq \frac{r^{(z_i)}_i}{\epsilon\cdot \gamma \log k + p} \geq \frac{r^{(z_i)}_i}{(1+\epsilon\cdot \gamma k\log k)\cdot p},
\end{align*}
where the first inequality holds since in $\epsilon$-large markets, we have $r^{(z)}_i \leq \epsilon\cdot \gamma \log k$, 
and the last inequality holds since $p\geq \frac{1}{k}$. 

In the decomposed instance, for any $z > z_i$, the maximum value for those decomposed agents is less than $p$, so they will never accept the price. 
Let $\hat{z}_i \leq z_i$ be the maximum index such that the maximum value for agent $\hat{z}_i$ is at least $p$. 
The total probability of all agents who are willing to accept price $p$ is at most 
\begin{align*}
\hat{A_i}(p)&\triangleq\sum_{z=1}^{\hat{z}_i} \frac{r^{(z)}_i-r^{(z-1)}_i}{r^{(z)}_i-r^{(z-1)}_i+(1-\quant^{(z)}_i+\quant^{(z-1)}_i)\cdot p}\\
&\leq \sum_{z=1}^{\hat{z}_i} \frac{r^{(z)}_i-r^{(z-1)}_i}{(1-\epsilon\cdot \gamma k\log k)\cdot p}
\leq \frac{r^{(z_i)}_i}{(1-\epsilon\cdot \gamma k\log k)\cdot p}.
\end{align*}
The first inequality holds since $r^{(z)}_i$ is increasing in $z$, and $\quant^{(z)}_i \leq \epsilon\cdot \gamma k\log k$ for any $z$ under the $\epsilon$-large market assumption. 

Now we bound the expected number of accepts in $k$-unit auctions. 
Fix any sequence of agents in the \apname~pricing mechanism, which does not affect the expected revenue. Let $Q_i(p)$ and $\hat{Q}_i(p)$ be the probabilities such that there exists at least one unit remaining for agent $i$ given price $p$ before and after splitting the agents. 
Note that $Q_i(p) \geq \hat{Q}_i(p)$ for any $i$ and any $p$. 
Letting $\hat{p}\geq \frac{1}{k}$ be the optimal price that maximize the revenue from \apname~pricing given $\hat{\dist}$, have
\begin{equation*}
\up(\hat{\dist}) \leq \hat{p}\cdot \sum_i \hat{Q}_i(\hat{p}) \cdot \hat{A}_i(\hat{p}) 
\leq \hat{p}\cdot \sum_i Q_i(\hat{p}) \cdot A_i(\hat{p}) \cdot \frac{1+\epsilon\cdot \gamma k\log k}{1-\epsilon\cdot \gamma k\log k} 
\leq \frac{1+\epsilon\cdot \gamma k\log k}{1-\epsilon\cdot \gamma k\log k} \cdot \up(\dist).\qedhere
\end{equation*}
\end{proof}

Combining \cref{lem:optimal_increase} and \cref{lem:uniform_pricing_decreasing}, the approximation ratio decreases by a multiplicative factor at most
\begin{equation*}
\frac{1+\epsilon\cdot \gamma k\log k}{(1-\epsilon\cdot \gamma k\log k)^2}. \qedhere
\end{equation*}
\end{proof}

\subsection{Order Statistics}
\label{subsec:order_statistics}
Given \cref{prop:triangle_is_worst_case}, it is without loss to focus on triangular distributions. 
Moreover, the entire instance can be represented as a function of the order statistics. 
We use $\os_j(\val)$ to denote the cumulative distribution function for the $j$th order statistic for any $j\in [k]$.

Recall that for each agent $i$ with a triangular value distribution, the distribution can be captured by two parameters: monopoly revenue $\mrev_i$ and monopoly quantile $\mquant_i$. 
Let $\mval_i = \frac{\mrev_i}{\mquant_i}$ be the monopoly value. 
Given any instance with triangular distributions, the instance can be captured by the \textit{cumulative monopoly revenue function} $R(x) \eqdef \sum_{i \in [n]: \mval_i \ge x} \mrev_i$ and the \textit{cumulative monopoly quantile function} $Q(x) \eqdef \sum_{i \in [n]: \mval_i \ge x} \mquant_i$. 
By definition, both $R(x)$ and $Q(x)$ are decreasing functions.
Moreover, under the large market assumption, for any $x\geq \frac{1}{k}$, we have
\begin{align*}
\log \os_{1}(x)
&= \sum_{i \in [n]: \mval_i \ge x} \log \frac{x}{x + \frac{\mrev_i}{1 - \mquant_i}}\\
&= \sum_{i \in [n]: \mval_i \ge x} -\frac{\mrev_i}{x} \cdot(1+O(\epsilon k^2\log k))
= -\frac{R(x)}{x}\cdot(1+O(\epsilon k^2\log k))
\end{align*}
where the second equality holds since under the $\epsilon$-large market assumption, we have $\mquant_i \leq \frac{1}{x} k\log k$ and $\mrev_i \leq \epsilon k\log k$ for any agent $i$ such that $\mval_i \geq x$. 
Therefore, in the large market limit where $\epsilon\to 0$, we have
\begin{align}\label{eq:calculation:2}
R(x) = x \cdot (-\ln \os_{1}(x)).
\end{align}
Moreover, the cumulative monopoly quantile function $Q(x)$ can also be pinned down as 
\begin{align}\label{eq:calculation:3}
Q(x) = \int_{x}^{+\infty} -\frac{R'(z)}{z} \d z.
\end{align}

In addition to the above characterization, the expected revenue of the \apname~pricing mechanism, given any price $\pay$, can also be expressed as a function of the order statistics as follows:
\begin{align}\label{eq:ap_os}
\up(\pay, \dist) = \pay\cdot \sum_{j\in[k]} (1-\os_j(\pay)).
\end{align}
These characterizations indicate the usefulness of order statistics in quantifying the worst case gap of \apname~pricing compared to the optimal revenue. 
The details of these analyzes are provided in the subsequent sections. 
In this section, we will show that in large markets, all the higher-order statistics can be uniquely determined by the first-order statistics.

We defined the approximate $j$th order statistic as\footnote{If $\os_1(\pay) = 0$, we define $\widehat{\os}_j(\pay) = 0$.} 
\begin{align*}
\widehat{\os}_j(\pay) = \os_1(\pay)\cdot \sum_{0\leq t\leq j-1} \frac{1}{t!} (-\ln \os_1(\pay))^t.
\end{align*}
\begin{proposition}\label{prop:order_statistics}
Given any $j\in[k]$, any parameter $\delta\leq \frac{1}{4j}$ and any price $\pay>0$,
if $\dist_i(\pay)\geq 1-\delta$ for all $i$,
we have 
\begin{align*}
\os_j(\pay) \in \sbr{(1-\delta j)\cdot \widehat{\os}_j(\pay), (1+2\delta j)\cdot\widehat{\os}_j(\pay)}.
\end{align*}
\end{proposition}
In a $\epsilon$-large market, for any $\pay\geq \frac{1}{k}$, 
we have $\dist_i(\pay)\geq 1 - \epsilon k^2\log k$. 
In the limit where $\epsilon \to 0$, we have 
$\os_j(\pay) = \widehat{\os}_j(\pay)$ for all $j\in[k]$. 

\begin{proof}[Proof of \cref{prop:order_statistics}]
If $\os_1(\pay)=0$, then by definition $\widehat{\os}_j(\pay)=0$ and also $\os_j(\pay)=0$,
so the claim holds trivially. In the rest of the proof, we assume $\os_1(\pay)>0$.
Moreover, since the bounds hold trivially for $j=1$, we focus on the case when $j\geq 2$.

Note that for any $j\in[k]$ and any $\pay > 0$, 
\begin{align}
\os_j(\pay) &= \sum_{0\leq t\leq j-1} \sum_{|W|=t} \rbr{\prod_{z\not\in W} \dist_z(\pay)}
\cdot \rbr{\prod_{z\in W} (1-\dist_z(\pay))} \nonumber\\
&= \os_1(\pay) \cdot \sum_{0\leq t\leq j-1} \sum_{|W|=t} \rbr{\prod_{z\in W} \rbr{\frac{1}{\dist_z(\pay)}-1}}. \label{eq:osj_expression}
\end{align}
where the last equality holds by observing that $\prod_{i\in [n]} \dist_z(\pay) = \os_1(\pay)$.

\paragraph{Lower bound on $\os_j(\pay)$.} 
We first show that $\frac{1}{\dist_z(\pay)} - 1 \geq -\ln \dist_z(\pay)$. This is because, for function $g(x) = \frac{1}{x}-1+\ln x$, we have $g(1) = 0$ and $g'(x) = \frac{1}{x}-\frac{1}{x^2} \leq 0$ for any $x\in(0,1)$. 
The desirable inequality holds by observing that $\dist_z(\pay) \leq 1$ for any $\pay$. 
Therefore, we have 
\begin{align}\label{eq:lb1}
\os_j(\pay) \geq \os_1(\pay) \cdot \sum_{0\leq t\leq j-1} \sum_{|W|=t} \rbr{\prod_{z\in W} \rbr{-\ln \dist_z(\pay)}}.
\end{align}
Let $\Lambda = \sum_{i\in[n]} -\ln \dist_i(\pay)$
and let $P_i = \frac{-\ln \dist_i(\pay)}{\Lambda}$. 
Consider the process of drawing $t$ times with replacement from $n$ elements, where the probability that element $i$ is chosen is $P_i$. 
The probability of all $t$ elements being distinct is 
\begin{align*}
P_{\mathrm{DISTINCT}}(t)\triangleq t! \cdot \sum_{|W|=t} \prod_{z\in W} P_i
= \frac{t!}{\Lambda^t} \sum_{|W|=t} \prod_{z\in W} -\ln \dist_z(\pay).
\end{align*}
Substituting in \Cref{eq:lb1} and observing that $\Lambda = -\ln \os_1(\pay)$, we have 
\begin{align*}
\os_j(\pay) - \widehat{\os}_j(\pay) \geq 
\os_1(\pay) \cdot \sum_{0\leq t\leq j-1} \frac{\Lambda^t}{t!}(P_{\mathrm{DISTINCT}}(t)-1).
\end{align*}
Next, we upper bound the probability $P_{\mathrm{COLLISION}}(t) = 1- P_{\mathrm{DISTINCT}}(t)$. 
It is easy to verify that $P_{\mathrm{COLLISION}}(t)=0$ for $t=0$ or $1$. 

For $t\geq 2$, we first bound the probabilities $P_i$. Note that since $\delta\leq \frac{1}{2}$ and $\dist_i(\pay) \geq 1-\delta$ for any $i\in[n]$, we have 
$-\ln \dist_i(\pay) \leq \frac{1}{\dist_i(\pay)} - 1 \leq \frac{\delta}{1-\delta} \leq 2\delta$. 
Thus, $\max_{i\in[n]}P_i \leq \frac{2\delta}{\Lambda}$.
By the union bound over $\binom{t}{2}$ pairs of draws for collision, we have 
\begin{align*}
P_{\mathrm{COLLISION}}(t) \leq \binom{t}{2} \cdot \sum_{i\in[n]} P_i^2 \leq \frac{\delta t(t-1)}{\Lambda} 
\end{align*}
where the last inequality holds since $\sum_{i\in[n]} P_i^2 \leq \max_{i\in[n]} P_i \cdot \sum_{i\in[n]} P_i = \max_{i\in[n]} P_i \leq \frac{2\delta}{\Lambda}$.
Next, we divide the analysis further into two cases. 
\begin{itemize}
\item $\Lambda \geq j$. In this case, we have $P_{\mathrm{COLLISION}}(t) \leq \delta j$ for any $t\leq j-1$, and hence
\begin{align*}
\os_j(\pay) - \widehat{\os}_j(\pay) \geq - \delta j \cdot \os_1(\pay) \cdot \sum_{0\leq t\leq j-1} \frac{\Lambda^t}{t!}
= - \delta j \cdot \widehat{\os}_j(\pay).
\end{align*}

\item $\Lambda \leq j$. In this case, since $P_{\mathrm{COLLISION}}(t)=0$ for $t=0$ or $1$, we have 
\begin{align*}
\os_j(\pay) - \widehat{\os}_j(\pay) \geq 
- \os_1(\pay) \cdot \sum_{2\leq t\leq j-1} \frac{\Lambda^t}{t!} \cdot \frac{\delta t(t-1)}{\Lambda}
= - \delta \Lambda \cdot \os_1(\pay) \cdot \sum_{0\leq t\leq j-3} \frac{\Lambda^t}{t!}
\geq - \delta j \cdot \widehat{\os}_j(\pay),
\end{align*}
where the last inequality holds since $\Lambda \leq j$ and $\Lambda > 0$. 
\end{itemize}
Combining the two cases, we have the desirable lower bound of
\begin{align*}
\os_j(\pay) \geq (1-\delta j)\cdot \widehat{\os}_j(\pay).
\end{align*}

\paragraph{Upper bound on $\os_j(\pay)$.} 
In this case, we first prove that 
\begin{align}\label{eq:upperbound_ln}
\frac{1}{\dist_z(\pay)} - 1 \leq - (1+\delta)\cdot \ln \dist_z(\pay).
\end{align}
First note that since $\dist_z(\pay) \geq 1-\delta$, we have
\begin{align*}
\frac{1-\dist_z(\pay)}{\dist_z(\pay)} - \bigl(-\ln \dist_z(\pay)\bigr)
= \int_0^{1-\dist_z(\pay)} \frac{x}{(1-x)^2}\,dx
\leq 2\int_0^{1-\dist_z(\pay)} x\,dx
= (1-\dist_z(\pay))^2
\leq - \delta \cdot \ln \dist_z(\pay),
\end{align*}
where the first inequality holds since $(1-x)^2 \geq (1-\delta)^2 \geq \frac{1}{2}$ for $x\leq 1-\dist_z(\pay) \leq \delta$ and $\delta \leq \frac{1}{4}$, 
and the second inequality holds since $1-\dist_z(\pay) \leq \delta$ and $1-\dist_z(\pay) \leq - \ln \dist_z(\pay)$ for $\dist_z(\pay) \in (0,1]$.
By rearranging the terms, \Cref{eq:upperbound_ln} holds. 
By applying \Cref{eq:upperbound_ln} in \Cref{eq:osj_expression}, we have
\begin{align*}
\os_j(\pay) &\leq \os_1(\pay) \cdot \sum_{0\leq t\leq j-1} (1+\delta)^t \cdot \sum_{|W|=t} \rbr{\prod_{z\in W} \rbr{-\ln \dist_z(\pay)}}\\
&= \os_1(\pay) \cdot \sum_{0\leq t\leq j-1} (1+\delta)^t \cdot  \frac{\Lambda^t}{t!}\cdot P_{\mathrm{DISTINCT}}(t)
\leq \os_1(\pay) \cdot \sum_{0\leq t\leq j-1} (1+\delta)^t \cdot \frac{\Lambda^t}{t!}.
\end{align*}
Here the equality holds by the same argument in the lower bound part, and the last inequality holds since $P_{\mathrm{DISTINCT}}(t) \leq 1$ for any $t$. 
Since $\delta \leq \frac{1}{4j}$, we have $\delta t\leq \delta j \leq \frac{1}{4}$, and hence 
\begin{align*}
(1+\delta)^t \leq e^{\delta t}\leq 1+2\delta t
\leq 1+2\delta j.
\end{align*}
Combining the inequalities, we have
\begin{equation*}
\os_j(\pay) \leq (1+2\delta j)\cdot \os_1(\pay) \cdot \sum_{0\leq t\leq j-1} \frac{\Lambda^t}{t!} 
= (1+2\delta j)\cdot \widehat{\os}_j(\pay). \qedhere
\end{equation*}
\end{proof}

\subsection{Asymptotic Upper Bounds}
\label{subsec:bound_calculation}

In this subsection, we will show that
\begin{align*}
    \opt(F^{*})
    \le \frac{2}{1 - k^{k} / (k! \cdot e^{k})}
    \le \frac{2}{1 - 1 / \sqrt{2\pi k}}.
\end{align*}
The last inequality applies Stirling's approximation.

\begin{lemma}\label{lem:worst_os}
The first-order statistic of the worst case instance is given by the following implicit function:
\begin{align}\label{eq:calculation:4}
1 = x \cdot \left(k - \sum_{t \in [0 : k - 1]} \frac{k - t}{t!} \cdot \os_{1}(x) \cdot (-\ln \os_{1}(x))^{t}\right), \quad\forall x\geq \frac{1}{k}.
\end{align}
\end{lemma}
\begin{proof}[Proof of \cref{lem:worst_os}]
Subject to the constraint that $\up(\dist) = 1$, the worst case instance satisfies the condition that 
$\up(\pay,\dist) = 1$ for all $p\geq \frac{1}{k}$. 
This is because if there exists $p \geq \frac{1}{k}$ such that $\up(\pay,\dist) < 1$, we can increase the first order statistic at $p$ such that the optimal revenue increases, without violating the condition that $\up(\dist) = 1$.
Therefore, for any $x\geq \frac{1}{k}$, we have
\begin{align*}
1 &= \up(x,\dist) = x \cdot \sum_{j \in [1 : k]} (1 - \os_{j}(x)) 
= x \cdot \left(k-\sum_{j \in [1 : k]} \sum_{t \in [0: j - 1]} \frac{1}{t!} \cdot \os_{1}(x) \cdot (-\ln \os_{1}(x))^{t}\right)\\
& = x \cdot \left(k - \sum_{t \in [0 : k - 1]} \frac{k - t}{t!} \cdot \os_{1}(x) \cdot (-\ln \os_{1}(x))^{t}\right)
\end{align*}
where the third equality holds by \cref{prop:order_statistics}, 
and the last equality holds by rearranging the terms.
\end{proof}

Combining \Cref{eq:calculation:2,eq:calculation:3,eq:calculation:4}, we can infer that the first-order value CDF $\os_{1}(x)$ has the support $\supp(\os_{1}) = [\frac{1}{k}, +\infty]$. Moreover, combining them with L'H\^{o}pital's rule further implies that
\begin{align*}
    & R(+\infty) = 1,
    \qquad
    Q(+\infty) = 0,\\
    & R(1 / k) = +\infty,
    \qquad
    Q(1 / k) = +\infty.
\end{align*}

\paragraph{Ex-Ante Relaxation.}
The notion of {\ExAnteRelaxation} $\EAR(F)$ was introduced in \cite{CHK07} to provide a tractable upper bound for the optimal revenue and was then extensively studied in \cite{A14,Y11,AHNPY19,JJLZ22}. 
Specifically, it relaxes the ex post feasibility constraint of selling at most $k$-units, to the ex ante constraint that the total probability of winning the items from all agents cannot exceed $k$. 
That is, letting $r_i(q_i)$ be the optimal revenue from agent $i$ given distribution $F_i$ subject to the constraint that the probability agent $i$ wins a unit is $q_i$, we have 
\begin{align*}
\EAR(F) = \max_{q: \sum_{i\in[n]} q_i\leq k} \sum_{i\in[n]} r_i(q_i).
\end{align*}
By definition, for any distribution $\dist$, we have 
\begin{align*}
    \EAR(F) \geq \opt(F).
\end{align*}

In the $k$-unit setting, when focusing on instances with triangular distributions, the ex ante relaxation takes a particularly simple form. Specifically,\footnote{If a triangular instance $F$'s total quantile is less than $k$, i.e., $\sum_{i \in [n]} q_{i} < k$, then $Q^{-1}(k)$ is not well-defined; instead, $\EAR(F) \eqdef \sum_{i \in [n]} v_{i} q_{i}$. (However, for the worst-case instance $F^{*}$, this issue does not arise.)}
\begin{align*}
    \EAR(F) \eqdef R(Q^{-1}(k)).
\end{align*} 
The intuition of this representation is simple. 
To maximize $\sum_{i\in[n]} r_i(q_i)$ subject to an ex ante constraint on selling probabilities, it is optimal to sell to agents with the highest $\mrev_i = \frac{\mrev_i}{\mquant_i}$. 
The cumulative monopoly revenue function $R$ exactly captures the revenue contribution from those agents.

\paragraph{Asymptotic Upper Bounds.}
To upper bound $\EAR(F^{*})$, we first introduce two auxiliary parameters:
\begin{align}
\alpha &\eqdef Q^{-1}(k), \label{eq:calculation:5}\\
\beta &\eqdef \os_{1}^{-1}(e^{-k}). \label{eq:calculation:6}
\end{align}
Note that by \Cref{eq:calculation:2}, we have 
\begin{align}\label{eq:Rbeta}
R(\beta) = \beta \cdot (-\ln \os_{1}(\beta)) = k\beta.
\end{align}

Next, we deduce a $\beta$-based upper-bound for $\EAR(F^{*})$.
Note that 
\begin{align*}
k \overset{\eqref{eq:calculation:5}}{=} Q(\alpha)
\overset{\eqref{eq:calculation:3}}{=} \int_{\alpha}^{+\infty} \frac{(-R'(z))}{z} \cdot \d z
\ge \int_{\alpha}^{\beta} \frac{(-R'(z))}{\beta} \cdot \d z
= \frac{R(\alpha) - R(\beta)}{\beta}.
\end{align*}
The inequality holds since in the case where $\beta < \alpha$, we observe that $\LHS \ge 0 \ge \RHS$ as $R(x)$ is a decreasing function. 
The case where $\beta \ge \alpha$ is obvious.
This implies that 
\begin{align}
\EAR(F^{*})
= R(\alpha) \le k \beta + R(\beta)
\overset{\eqref{eq:Rbeta}}{=} 2k \beta.\label{eq:calculation:7}
\end{align}

Secondly, we actually have an explicit formula for $\beta$:
\begin{align}
    \beta = \frac{1 / k}{1 - k^{k} / (k! \cdot e^{k})}
    \le \frac{1 / k}{1 - 1 / \sqrt{2\pi k}}.
    \label{eq:calculation:8}
\end{align}
Here is the detailed deduction:
\begin{align*}
    1
    & \overset{\eqref{eq:calculation:4}, \eqref{eq:calculation:6}}{=}
    \beta \cdot \left(k - \sum_{t \in [0 : k - 1]} \frac{k - t}{t!} \cdot e^{-k} \cdot k^{t}\right)\\
    & = \beta k \cdot \left(1 - \sum_{t \in [0 : k - 1]} \frac{1}{t!} \cdot e^{-k} \cdot k^{t} + \sum_{t \in [1 : k - 1]} \frac{1}{(t - 1)!} \cdot e^{-k} \cdot k^{t - 1}\right)\\
    & = \beta k \cdot \left(1 - \frac{e^{-k} \cdot k^{k - 1}}{(k - 1)!}\right)\\
    & = \beta k \cdot \left(1 - \frac{k^{k}}{k! \cdot e^{k}}\right)
\end{align*}

Combining \Cref{eq:calculation:7,eq:calculation:8}, we have
\begin{align}
    \opt(F^{*})
    \le \EAR(F^{*})
    \le \frac{2}{1 - k^{k} / (k! \cdot e^{k})}
    \le \frac{2}{1 - 1 / \sqrt{2\pi k}}.
    \label{eq:UB-asymptotic}
\end{align}

\subsection{Universal Upper Bounds}
\label{subsec:universal_bound}

Let us rewrite $F^{*(k)} = F^{*}$ to emphasize the specific choice of $k \in \NN$. In this subsection, we will show that
\begin{align}
    & \opt(F^{*(k)}) \le \opt(F^{*(1)}) \approx 2.4762,
    && \forall k \in \NN.
    \label{eq:UB-universal}
\end{align}

\paragraph{Case~$k = 1$.}
We can formulate $\os_{1}(x)$, $R(x)$, and $Q(x)$ explicitly, as follows:
\begin{align*}
    \os_1(x) & \overset{\eqref{eq:calculation:4}}{=} 1 - \tfrac{1}{x}\\
    R(x) & \overset{\eqref{eq:calculation:2}}{=} -x \ln( 1 - \tfrac{1}{x}) = \sum_{t = 1}^{+\infty} \tfrac{1}{t} \cdot x^{-(t - 1)}\\
    Q(x) & \overset{\eqref{eq:calculation:3}}{=} \int_{x}^{+\infty} \frac{(-R'(z))}{z} \cdot \d z = \sum_{t = 1}^{+\infty} \frac{t - 1}{t^{2}} x^{-t}
\end{align*}
The cumulative distribution function of the \textit{first-order virtual value} $\Phi_{1}(x)$ satisfies the following identity:
\begin{align*}
    \Phi_{1}(x)
    = \prod_{i \in [n]: v_{i} \ge x} (1 - q_{i})
    %\overset{L.M.}{\longrightarrow} 
    = \prod_{i \in [n]: v_{i} \ge x} e^{-q_{i}}
    = e^{-Q(x)}
\end{align*}
where the first inequality holds for any triangular instance $F$, and the second inequality holds under the large market assumption. 
Moreover, for $x \in [0, 1]$, we haver $Q(x) = +\infty$, which further implies that $\Phi_{1}(x) = 0$.

Thus, the optimal revenue $\opt(F^{*(1)})$ is given by
\begin{align*}
    \opt(F^{*(1)})
    & = \int_{0}^{+\infty} (1 - \Phi_{1}(x)) \cdot \d x + R(+\infty)\\
    & = 2 + \int_{1}^{+\infty} (1 - e^{-Q(x)}) \cdot \d x\\
    & \approx 2.4762.
\end{align*}

\paragraph{Case~$k = 2$.}
Rearranging \Cref{eq:calculation:4} gives
\begin{align}
    & x = \frac{1}{2 - 2\os_{1} - \os_{1} \cdot (-\ln \os_{1})}
    \label{eq:calculation:9}\\
    \implies\quad
    & \frac{\d x}{\d \os_{1}}
    = \frac{1 + (-\ln \os_{1})}{(2 - 2\os_{1} - \os_{1} \cdot (-\ln \os_{1}))^{2}}.
    \label{eq:calculation:10}
\end{align}
Then, we can deduce the following identify for the functions $Q(x)$ and $\os_{1}(x)$:
\begin{align*}
    Q
    & = \int_{x}^{+\infty} \frac{(-R'(z))}{z} \cdot \d z\\
    & \overset{\eqref{eq:calculation:2}}{=} \int_{x}^{+\infty} \frac{z \cdot \frac{\os'_{1}(z)}{\os_{1}} - (-\ln \os_{1})}{z} \cdot \d z\\
    & = \big(\ln \os_{1}\big)\big|_{x}^{+\infty} - \int_{x}^{+\infty} \frac{(-\ln \os_{1})}{z} \cdot \d z\\
    & = (-\ln \os_{1}) - \int_{x}^{+\infty} \frac{(-\ln \os_{1})}{z} \cdot \frac{\d z}{\d \os_{1}} \cdot \d \os_{1}\\
    & \overset{\eqref{eq:calculation:9}, \eqref{eq:calculation:10}}{=} (-\ln \os_{1}) - \int_{x}^{+\infty} \frac{(-\ln \os_{1}) \cdot (1 + (-\ln \os_{1}))}{2 - 2\os_{1} - \os_{1} \cdot (-\ln \os_{1})} \cdot \d \os_{1}.
\end{align*}
In this format, we can numerically obtain $\alpha = Q^{-1}(2) \approx 0.5206$, $\os_{1}(\alpha) \approx 0.012390$, and
\begin{align*}
    \EAR(F^{*(2)}) = R(\alpha) \approx 2.2860.
\end{align*}

\paragraph{Cases~$k = 3, 4$.}
By adapting the analysis in the case $k = 2$ to the cases $k = 3, 4$, we can deduce that $\opt(F^{*(3)}) \le \EAR(F^{*(3)}) \approx 2.1914$ and $\opt(F^{*(4)}) \le \EAR(F^{*(4)}) \approx 2.1432$.

\paragraph{Cases~$k \ge 5$.}
Following \Cref{eq:UB-asymptotic}, we have $\opt(F^{*}) \le \frac{2}{1 - 1 / \sqrt{10\pi}} \approx 2.4343$.

\medskip
\noindent Combining all cases, the worst case approximation ratio occurs at $k=1$, with the worst case ratio being $\approx 2.4762$.

\bibliographystyle{apalike}
\bibliography{reference}

\appendix
\section{Missing Proofs}
\label{apx:proofs}

\begin{proof}[Proof of \cref{cor:loose_gap}]
The lower bound is immediate, as regular distributions are special cases of quasi-regular distributions. 
Next, we prove the upper bound. 

First, when the value distributions are quasi-regular, the worst case approximation between \sppname~pricing and \apname~pricing occurs with triangular distributions. 
Moreover, since triangular distributions are regular, \cref{prop:loose_gap} implies that the worst case gap between \sppname~pricing and \apname~pricing is at most $O(\log k)$. 
Finally, since the worst case gap between the optimal and \sppname~pricing is $1+\frac{\Theta(1)}{\sqrt{k}}$ \citep{Y11}, 
the worst case gap between the optimal and \apname~pricing is at most $O(\log k)$.  
\end{proof}

\begin{proof}[Proof of \cref{cor:large_cdf}]
Suppose by contradiction that there exists a price $\pay\geq \frac{1}{k}$ and an agent~$i$ such that $\dist_i(\pay) < 1 - \epsilon \cdot\gamma k\log k$.
This implies that the expected revenue from agent $i$ is strictly larger than $\epsilon \cdot\gamma \log k$ by posting a price $\pay\geq \frac{1}{k}$. 
This contradicts to the $\epsilon$-large markets assumption. 
\end{proof}

\subsection{Lower Bounds}
\label{subapx:lowerbound}
In this section, we show that the worst case gap between \apname~pricing and the optimal revenue is at least 2, even under the large market assumption. 
Specifically, consider the following instance, parameterized by a sufficiently small constant $\delta \in (0,1)$.
The agents are divided into two groups. 
In the first group, there are $C_{\delta}$ agents. 
Each agent in this group has a triangular value distribution 
$\tri_{\delta,\delta^2}$. Note that these agents have a monopoly reserve $\frac{\delta}{\delta^2}=\frac{1}{\delta}$, 
and the parameter $C_{\delta}$ is chosen such that by posting a price $\frac{1}{\delta}$ to those agents, the expected number of items sold to those agents is $\delta$. 
In the second group, there are $\frac{1}{\delta^{10}}$ agents. 
Each agent in this group has a triangular value distribution 
$\tri_{\frac{\delta}{k},\delta}$. 
The monopoly reserve for those agents is $\frac{1}{k}$. 
Moreover, by posting a price of $\frac{1}{k}$ to those agents, the probability that all items are sold converges to $1$ as $\delta\to 0$. 

In this constructed instance, the revenue from anonymous pricing is at most $1$. Moreover, when $\delta\to 0$, 
the optimal revenue converges to $2$ by posting a price of $\frac{1}{\delta}$ to the first group, and selling the remaining items to the second group at a price of $\frac{1}{k}$. 
The revenue gap is $2$ in the limit.

\end{document}